\documentclass[copyright,creativecommons]{eptcs}
\providecommand{\event}{GandALF 2017} 

\usepackage[utf8]{inputenc}
\usepackage[T1]{fontenc}
\usepackage[draft]{fixme}\fxusetheme{color}
\usepackage{paralist}
\usepackage{booktabs}
\usepackage{amsmath,amssymb}
\usepackage{mathtools}
\usepackage{colortbl}
\usepackage{multirow}
\usepackage{enumerate}

\usepackage{tikz}
\usetikzlibrary{fit,calc,positioning,decorations.pathreplacing,matrix}

\usepackage{amsthm}
\newtheorem{theo}{Theorem}

\newtheorem{cor}[theo]{Corollary}
\newtheorem{prop}[theo]{Proposition}
\newtheorem{mydef}[theo]{Definition}

\newtheorem{lemme}[theo]{Lemma}

\newcommand{\R}{\mathbb{R}}

\newcommand{\N}{\mathbb{N}}

\newcommand{\LazyImp}{Lazy Improvement}

\newcommand{\SIDyn}{$\{SI\}$-Dynamics}
\newcommand{\LDyn}{$\{I, L, 1P\}$-Dynamics}
\newcommand{\ADyn}{$\{SI, A\}$-Dynamics}
\newcommand{\CDyn}{$\{I, L\}$-Dynamics}

\newcommand{\IADyn}{$\{I, A\}$-Dynamics}

\newcommand{\dynamics}{\rightharpoonup}
\newcommand{\dynamicsG}{\underset{G}{\rightharpoonup}}

\newcommand{\ssprim}{(s,s')}

\newcommand{\dyn}[1]{\ensuremath{\xrightharpoonup{#1}}}

\newcommand{\Ldyn}{\dyn{\{I, L, 1P\}}} 
\newcommand{\Odyn}{\dyn{SI, A}}
\newcommand{\Adyn}{\Odyn}
\newcommand{\SIdyn}{\dyn{SI}}

\newcommand{\IAdyn}{\dyn{\{I, A\}}}

\newcommand{\Crazydyn}{\rightsquigarrow}

\newcommand{\CEq}{SNE}

\newcommand{\wo}{\lesssim}

\newcommand{\lo}{<}
\newcommand{\eq}{\sim}

\newcommand{\Hs}[1]{H(#1)}
\newcommand{\Hsi}[1]{H_i(#1)}

\newcommand{\StratG}{Strat_G}

\newcommand{\statesof}[1]{\ensuremath{H_{#1}}}
\newcommand{\outcome}[1]{\ensuremath{\left\langle #1\right\rangle}}

\title{Dynamics and Coalitions in Sequential Games}
\author{Thomas Brihaye %
  \institute{UMONS, Belgium} %
  \email{thomas.brihaye@umons.be}
  \and
  Gilles Geeraerts\thanks{%
    This author was supported by an
    \emph{Advanced A.R.C.} grant `Non-Zero Sum Game Graphs:
    Applications to Reactive Synthesis and Beyond' from the
    Fédération Wallonie-Bruxelles.%
  }%
  \institute{Université libre de Bruxelles, Belgium} %
  \email{gigeerae@ulb.ac.be}
  \and
  Marion Hallet\thanks{%
    This author was supported by an F.R.S./FNRS `aspirant'
    fellowship.%
  }%
  \institute{UMONS, Belgium} %
  \email{marion.hallet@umons.be}
  \and
  Stéphane Le Roux\thanks{%
    This author was supported by the ERC inVEST (279499)
    project.%
  }%
  \institute{Université libre de Bruxelles, Belgium} %
  \email{sleroux@ulb.ac.be}
}

\def\titlerunning{Dynamics and Coalitions in Sequential Games}
\def\authorrunning{T. Brihaye, G. Geeraerts, M. Hallet \& S. Le
  Roux}
\def\copyrightholders{T. Brihaye, G. Geeraerts,\\ M. Hallet \& S. Le
  Roux}

\begin{document}

\maketitle

\begin{abstract}
  We consider $n$-player non-zero sum games played on finite trees
  (i.e., sequential games), in which the players have the right to
  repeatedly update their respective strategies (for instance, to
  improve the outcome wrt to the current strategy profile). This
  generates a dynamics in the game which may eventually stabilise to a
  Nash Equilibrium (as with Kukushkin’s lazy improvement), and we
  argue that it is interesting to study the conditions that guarantee
  such a dynamics to terminate.

  We build on the works of Le Roux and Pauly who have studied
  extensively one such dynamics, namely the Lazy Improvement
  Dynamics. We extend these works by first defining a turn-based
  dynamics, proving that it terminates on subgame perfect equilibria,
  and showing that several variants do not terminate. Second, we
  define a variant of Kukushkin’s lazy improvement where the players
  may now form coalitions to change strategies. We show how properties
  of the players' preferences on the outcomes affect the termination
  of this dynamics, and we thereby characterise classes of games where it
  always terminates (in particular two-player games).
\end{abstract}

\section{Introduction}

Since the seminal works of Morgenstern and von Neuman in the forties
\cite{vNM44}, game theory has emerged as a prominent paradigm to model
the behaviour of rational and selfish agents acting in a competitive
setting. The first and main application of game theory is to be found
in the domain of economics where the agents can model companies or
investors who are competing for profits, to gain access to market,
etc. Since then, game theory has evolved into a fully developed
mathematical theory and has recently found many applications in
computer science. In this setting, the agents usually model different
components of a computer system (and of its environment), that have
their own objective to fulfil. For example, game theory has been
applied to analyse peer-to-peer file transfer protocols \cite{Nisan}
where the agents want to maximise the amount of downloaded data in
order to obtain, say, a given file; while minimising the amount of
uploaded data to save bandwidth. Another application is that of
controller synthesis where the two, fully antagonistic, agents are the
system and its environment, and where the game objective models the
control objective.

The most basic model to describe the interaction between the players
is that of games in \emph{strategic form} (aka matrix games), where
all the players chose simultaneously an action (or strategy) from a
finite set, and where they all obtain a payoff (usually modelled by a
real value) which depends on their joint choice of actions (or
strategy profile).  Strategic games are one-shot games, in the sense
that the players play only one action, simultaneously. An alternative
model where players play in turn is that of
\emph{sequential games}, which we consider in this work. Such a game
is played on a finite tree whose inner nodes are labelled by players,
whose edges correspond to the possible actions of the players, and
whose leaves are associated with \emph{outcomes}. The game starts at
the root, and, at each step of the game, the player who owns the
current node picks an edge (i.e. an action) from this node, and the
game moves to the destination node of the edge. The game ends when a
leaf is reached, and the players obtain the outcome associated with
that leaf. Moreover, each player has a \emph{preference} over the
possible outcomes, and each player's aim is thus to obtain the best
outcome according to this preference relation.

Arguably the most natural question about these games is to determine
how rational and selfish players would act (assuming that they have
full knowledge of the other players's possible actions and of the
outcomes). Several answers to this question have been provided in the
literature, such as the famous notion of Nash Equilibrium
\cite{Nash50}, which is a situation (strategy profile) in which no
player has an incentive to change his choice of action alone (because
such a choice would not be profitable to him). Apart from Nash
equilibria, other \emph{solution concepts} have been proposed like
Subgame Perfect Equilibria. This traditional view on game theory can
be qualified as \emph{static}, in the sense that all players chose
their strategies (thereby forming a strategy profile that can qualify
as one of the equilibria listed above), and the game is then played
\emph{once} according to these strategies. It is, however, also
natural to consider a more \emph{dynamic} setting, in which the
players can play the game repeatedly, updating their respective
strategies after each play, in order to try and improve the outcome at
the next play.

\paragraph{Contribution} In this paper, we continue a series of works
that aim at connecting these static and dynamic views. That is, we
want to study (in the case of extensive form games) the long term
behaviour of the dynamics in which players update repeatedly their
strategies and characterise when such dynamics converge, and to what
kind of strategy profiles (i.e., do these stable profiles correspond
to some notions of equilibria?). Obviously, this will depend on how
the players update their strategies between two plays of the game. Our
results consist in identifying minimal conditions on the updates
(modelling potential rational behaviours of the players) for which we
can guarantee convergence to some form of equilibria after a bounded
number of updates. Our work is an extension of a previous paper by Le
Roux and Pauly \cite{LeRouxDynamics}, where they study extensively the
so-called \emph{Lazy Improvement Dynamics}. Intuitively, in this
dynamics, a single player can update his strategy at a time, in order
to improve his outcome, and \emph{only} in nodes that have an
influence on the final outcome (lazy update). Their main result is
that this dynamics terminate on Nash equilibria when the preferences
of the players are acyclic.  Our contribution consists in considering
a broader family of dynamics and characterising their
termination. More precisely:
\begin{itemize}
\item We start (Section~\ref{sec:subg-impr-dynam}) by considering all
  dynamics that respect the \emph{subgame improvement} property, where
  players update their strategies only if this yields a better outcome
  in all subgames that are rooted in a node where a change has been
  performed. We argue that this can be regarded as a rational
  behaviour of the players: improving in the subgames can be regarded
  as an incentive, even when this does not improve the global outcome
  of the game. Indeed, such an improvement can turn out to be
  profitable if one of the other players later deviates from its
  current choice (this is the same intuition as the one behind the
  notion of Subgame Perfect Equilibrium).  Note that such dynamics
  have not been considered at all in \cite{LeRouxDynamics}. We show
  that, in all games where the preferences of the players are acyclic,
  these dynamics terminate and the terminal profiles are exactly the
  Subgame Perfect Equilibria of the game.
\item Then, we consider (Section~\ref{sec:cimp-dynamics}) all the
  dynamics that respect the \emph{improvement} property, where all
  players that change their respective strategy improve the outcome
  (from the point of view of their respective preferences). Among
  these dynamics are several ones that have already been studied by Le
  Roux and Pauly \cite{LeRouxDynamics} such as the Lazy Improvement
  Dynamics. We complete the picture (see
  \tablename~\ref{tab:results-I}), in particular we consider the
  dynamics that satisfies the \emph{improvement} and the
  \emph{laziness} properties but does not restrict the update to be
  performed by a single player, contrary to the \emph{Lazy Improvement
    Dynamics} of Le Roux and Pauly. Thus in our dynamics, players play
  lazily but are allowed to form coalitions to obtain an outcome which
  is better for all the players of the coalition. We give necessary
  and sufficient conditions on the preferences of the players, to
  ensure termination (on Strong Nash Equilibria), in several classes
  of games (among which 2 player games).
\end{itemize}

\paragraph{Related works} The most related works is the paper by Le
Roux and Pauly \cite{LeRouxDynamics} that we extend here, as already
explained. This work is inspired by the notions of \emph{potential}
and \emph{semi-potential} introduced respectively by Monderer and
Shapley \cite{MONDERER1996124}; and Kukushkin
\cite{KUKUSHKIN2002306}. Note also that the idea of repeatedly playing
a game and updating the players strategies between two plays is also
behind evolutionary game theory, but in this case, the rules governing
the updates are inspired from Darwinian evolution \cite{MP73,Wei95}.


\section{Preliminaries}
\label{sec:preliminaries}

\paragraph{Sequential games}
We consider \emph{sequential games}, which are $n$-player non-zero sum
games played on finite trees. \figurename~\ref{fig:twoplayer} shows an
example of such a game, with two players denoted $1$ and
$2$. Intuitively, each node of the tree is controlled by either of the
players, and the game is played by moving a token along the branches
of the tree, from the root node, up to the leaves, which are labelled
by a payoff (in this case, $x$, $y$ or $z$). The tree edges are
labelled by the actions that the player controlling the origin node can
play. For example, in the root node, Player $1$ can chose to play $r$,
in which case the game reaches the second node, controlled by Player
$2$, who can chose to play $l$. The payoff for both players is then
$y$. We also associate a preference relation with each player that
indicates how he ranks the payoffs. In the example of
\figurename~\ref{fig:twoplayer}, Player~$1$ prefers $z$ to $x$ and $x$
to $y$ (noted $y\prec_1 x \prec_1 z$), and Player~$2$ prefers $y$ to
$z$ and $z$ to~$x$.
Let us now formalise the basic notions about these games. The
definitions and notations of this section are inspired from
\cite{Osborne}.

\begin{mydef}\label{def:seqgame}
  A \emph{sequential} or \emph{extensive form} game $G$ is a tuple
  $\langle N, A, H, O, d, p, (\prec_i)_{i\in N} \rangle$ where:
  
  \begin{itemize}
  \item $N$ is a non-empty finite set of \textbf{players};
  \item $A$ is a non-empty finite set of \textbf{actions};
  \item $H$ is a finite set of finite sequences of $A$ which is
    prefix-closed. That is, the empty sequence $\varepsilon$ is a
    member of $H$; and $h = a^1\ldots a^k \in H$ implies that
    $h' = a^1 \ldots a^{\ell} \in H$ for all $\ell<k$.  Each member of
    $H$ is called a node. A node $h= a^1\ldots a^k\in H$ is
    \emph{terminal} if $\forall a\in A$, $a^1\ldots a^k a\notin
    H$. The set of terminal nodes is denoted by $Z$.
  \item $O$ is the non-empty set of \textbf{outcomes},
  \item $d : H\setminus Z \to N$ associates a player with each
    nonterminal node;
  \item $p : Z \to O$ associates an outcome with each terminal
    node;
  \item For all $i \in N$: $\prec_i$ is a binary
    relation over $O$,
    modelling the preferences of Player $i$. We write $x\prec_i y$ and
    $x\nprec_i y$ when $(x,y)\in \prec_i$ and $(x,y)\not\in\prec_i$
    respectively. 
  \end{itemize}
\end{mydef}

\begin{figure}
  \centering
  \begin{tikzpicture}
    
    \draw (2,3) node [circle, draw, inner sep=3pt] (a) {$1 $};
    \draw (1,2) node (aa) {$x $};
    \draw (3,2.) node  [circle, draw, inner sep=3pt] (ab) {$2 $};
    \draw (2,1) node  (abc) {$y $};
    \draw (4,1) node (aba) {$z $};
    
    \draw (1.3,2.6)  node {$ l $};
    \draw (2.7,2.6)  node {$ r $};
    \draw (2.3,1.6)  node {$ l $};
    \draw (3.7,1.6)  node {$ r $};

    \draw  (a) to (aa);
    \draw  (a) to (ab);
    \draw  (ab) to (abc);
    \draw  (ab) to (aba);
    
    \node[right of=ab, node distance=3cm]{
      $
      \begin{array}{c}
        y\prec_1 x \prec_1 z\\
        x\prec_2 z\prec_2 y
      \end{array}
      $
    } ;
  \end{tikzpicture}
  \caption{A sequential game with two players.}\label{fig:twoplayer}\label{fig:assocgraph}

\end{figure}
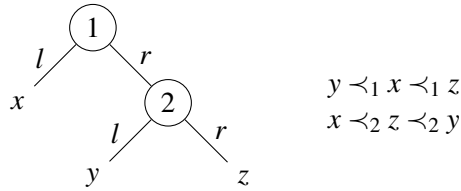

From now on, we fix a sequential game
$G=\langle N, A, H, O, d, p, (\prec_i)_{i\in N} \rangle$. Then, we let
$\statesof{i} = \{h \in H\setminus Z\mid d(h) =i\}$ be the set of
\textbf{nodes} belonging to player $i$.  A \textbf{strategy}
$s_i: \statesof{i} \to A $ of player $i$ is a function associating an
action with all nodes belonging to player $i$, s.t. for all
$h\in\statesof{i}$: $hs_i(h)\in H$, i.e., $s_i(h)$ is a legal action
from $h$. Then, a tuple $s = (s_i)_{i\in N}$ associating one strategy
with each player is called a \textbf{strategy profile} and we let
$\StratG$ be the set of all strategy profiles in $G$.  For all
strategy profiles $s$, we denote by
$\outcome{s}$ the \textbf{outcome} of $s$, which is the outcome of the
terminal node obtained when all players play according to
$s$. Formally, $\outcome{s}=p(h)$ where $h =a^1\ldots a^k\in Z$ is
s.t.  for all $0\le \ell < k $: $d(a^1\ldots a^\ell)=i$ implies
$a^{\ell+1}=s_i(a^1\ldots a^\ell)$.  Let $s = (s_i)_{i\in N}$ be a
strategy profile, and let $s^*_j$ be a Player $j$ strategy. Then, we
denote by $(s_{-j}, s_j^*)$ the strategy profile
$(s_1, ..., s_{j-1}, s^*_j, s_{j+1}, ..., s_{|N|})$ where all players
play $s_i$, except Player $j$ who plays $s^*_j$. Let $s$ be a strategy
profile, and let $h\in H\setminus Z$ be a nonterminal node. Then,
we let:
\begin{inparaenum}[(1)]
\item $H|_{h} = \{h' \mid hh' \in H\}$ be the \textbf{subtree} of $H$
  from $h$;
\item $s|_{h}$ be the \textbf{substrategy
  profile} of $s$ from $h$ which is s.t. $\forall hh'\in H\setminus Z$:
$s(hh')=s|_h(h')$; and
\item
  $G|_h =\langle N, A, H|_h, O, d|_h, p|_h, (\prec_i)_{i\in N}
  \rangle$ be the \textbf{subgame} of $G$ from $h$.  
\end{inparaenum}
Since a strategy profile $s=(s_i)_{i\in N}$ fixes a strategy for all
players, we abuse notations and write, for all nodes $h\in H$, $s(h)$
to denote the action $s_{d(h)}(h)$, i.e. the action that the owner of
$h$ plays in $h$ according to $s$.  Then, we say that a node
$h = a^1... ~a^k$ \textbf{lies along} the play induced by $s$ if $s(\varepsilon) = a^1$ and 
$\forall 1\le \ell<k$, $s(a^1\ldots a^{\ell})=a^{\ell+1}$.  As an
example, let us consider again the game in
\figurename~\ref{fig:twoplayer}. In this game, both players can chose
either $l$ or $r$ in the nodes they control. So, possible strategies
for Player $1$ and $2$ are $s_1$ s.t. $s_1(\varepsilon)=r$ and $s_2$
s.t. $s_2(r)=l$ respectively. Then, $\outcome{(s_1,s_2)}=y$. Observe
that, in our examples, we denote a profile of strategy by the actions
taken by the players. For example, we denote the profile
$s = (s_1,s_2)$ by $rl$. With this notation, this game has four
strategy profiles : $rr$, $rl$, $lr$ and $ll$.


\paragraph{Equilibria} Now that we have fixed the notions of games and
strategies, we turn our attention to three classical notions of
\emph{equilibria}. First, a strategy profile $s^*$ is a \textbf{Nash
  Equilibrium} (NE for short) if for all players $i\in N$, for all
strategies $s_i$ of player $i$:
$ \outcome{ (s^*_{-i}, s^*_i)} \nprec_i \outcome{ (s^*_{-i}, s_i)}.  $
It means that, in an NE $s^*$, no player has interest to deviate alone
from his current choice of strategy (because no such possible
deviation is profitable to him). On the other hand, a strategy profile
$s^*$ is a \textbf{Subgame Perfect Equilibrium} (SPE for short) if,
for all players $i\in N$, for all strategies $s_i$ of player $i$, for
all nonterminal nodes $h\in H_i$ of player $i$:
$ \outcome{(s^*_{-i}|_h, s^*_i|_h)}\nprec_i\outcome{(s^*_{-i}|_h,
  s_i|_h)}.  $ In other words, $s^*$ is an NE in every subgame of $G$.
Finally, in \cite{Aumann59}, Aumann defines a \textbf{Strong Nash
  Equilibrium} (SNE for short) as a strategy profile in which there is
no \emph{coalition} of players that has an incentive to deviate,
because such a deviation would allow \emph{all players} of the
coalition to improve the outcome.  Formally, a strategy profile $s^*$
is an SNE if, for all coalitions of players $I\subseteq N$, for all
strategy profiles $S_I$ of the coalition $I$, there is $i\in I$ s.t.:
$\outcome{(s^*_{-I}, s^*_I)}\not\prec_i\outcome{(s^*_{-I},s_I)}$.
Thus, the notion of SNE is stronger than the notion of NE. Actually,
the notion of SNE has sometimes been described as `too strong' in the
literature, because there are few categories of games in which SNEs
are guaranteed to exist (contrary to NEs and SPEs). This has prompted
other authors to introduce alternative solution concepts such as
Coalition Proof Equilibria \cite{BERNHEIM19871}.

For example, in the game in \figurename~\ref{fig:twoplayer}, the only
NE is $ll$. Moreover, it is also an SPE because $l$ is an NE in the
only subgame of $G$. However, $ll$ is not an SNE, because if both
players decide to form a coalition and play the $r$ action, they
obtain $z$ as outcome which is better than $x$ for both of them. There
is actually no SNE in this game. If we consider the same game with
following preferences for Player~$2$: $x\prec_2 y \prec_2 z$, then
$ll$ is still an NE, but not an SPE. The only SPE of this game is $rr$
(which is also an NE \emph{and} an SNE).

\paragraph{Dynamics: general definitions}
Let us now introduce the central notion of the work, i.e. dynamics in
extensive form games. As explained in the introduction, we want to
study the behaviour of the players when they are allowed to repeatedly
update their current strategy in a strategy profile, and characterise
the cases where such repeated updates (i.e. dynamics) converge to one
of the equilibria we have highlighted above. More specifically, we
want to characterise when a dynamics terminates for a game $G$, i.e.,
when players cannot infinitely often update their strategy.

Formally, for a sequential game $G$, a \textbf{dynamics} $\dynamicsG$
is a binary relation over $\StratG\times \StratG$, where, as usual, we
write $s \dynamicsG{} s'$ whenever $(s, s') \in
\dynamicsG$. Intuitively, $s\dynamicsG{} s'$ means that the dynamics
$\dynamicsG{}$ under consideration allows the strategy profile $s$ to
be updated into $s'$, by the change of strategy of $0$, $1$ or several
players. When the context is clear, we drop the name of the game $G$
and write $\dynamics$ instead of $\dynamicsG$. Given this definition,
it is clear that a dynamics $\dynamics$ corresponds to a directed
graph $(\StratG, \dynamics)$, where $\StratG$ is the set of graph
nodes, and $\dynamics$ is its set of edges. For example, the graphs in
\figurename~\ref{fig:graphs} are some possible graphs representing
dynamics associated to the game in \figurename~\ref{fig:twoplayer}.
Then, we say that the dynamics~$\dynamics$ \textbf{terminates} if
there is no infinite sequence of strategy profiles $(s^k)_{k\in \N}$
such that $\forall k \in \N$, $s^k \dynamics s^{k+1}$.  Equivalently,
$\dynamics$ terminates iff its corresponding graph is
acyclic. Intuitively, a dynamics terminates if players can not update
their strategy infinitely often, which means that the game will
eventually reach stability.  We say that a strategy profile $s$ is
\textbf{terminal} iff there is no $s'$ s.t. $s\dynamics s'$ (i.e., $s$
is a deadlock in the associated graph). Finally, given a pair of
strategies $s$ and $s'$, we let
$\Hs{s,s'} = \{h \in H| s(h) \neq s'(h)\}$ (resp.
$\Hsi{s,s'} = \{h \in H_i| s(h) \neq s'(h)\}$) be the set of nodes
(resp. the set of nodes belonging to player $i$) where the player who
owns the node plays differently according to $s$ and $s'$. Moreover,
$d(\Hs{s,s'})=\{d(h)\in N\mid h \in \Hs{s,s'}\} $ is the natural
extension of $d$ to subsets of $H$.

Let us now identify peculiar families of dynamics (of whom we want to
characterise the terminal profiles), by characterising how players
update their strategies from one profile to another.
\begin{mydef}[Properties of strategy updates]\label{def:properties}
  Let $s, s'$ be two strategy profiles of a game $G$. Then:
  \begin{enumerate}
  \item $\ssprim$ verifies the \textbf{Improvement Property}, written
    $\ssprim \models I$ if $\forall i\in N$: $s_i\neq s'_i$ implies
    $ \outcome{s} \prec_i\outcome{s'}$. That is, every player that
    changes his strategy improves his  payoff.
  \item $\ssprim$ verifies the \textbf{Subgame Improvement Property},
    written $\ssprim \models SI$, if $\forall i\in N$, $\forall h \in
    \Hsi{s,s'}$: $\outcome{s|_h } \prec_i \outcome{ s'|_h}$. That is,
    every player that changes his strategy improves his (induced)
    payoff in all the subgames rooted at one of the changes.
  \item $\ssprim$ verifies the \textbf{Laziness Property}, written
    $\ssprim \models L$, if $\forall h \in \Hs{s,s'}$, $h$ lies along
    the play induced by $s'$. Intuitively, we consider such updates as
    \emph{lazy} because we require that the players do not change
    their strategy in nodes which do not influence the payoff.
  \item $\ssprim$ verifies the \textbf{One Player Property}, written
    $\ssprim \models 1P$, if $\exists i\in N$ such that
    $\forall j \neq i$, $s_j = s_j'$. That is, at most one player
    updates its strategy (but he can perform changes in as many nodes
    as he wants).
  \item $\ssprim$ verifies the \textbf{Atomicity Property}, written
    $\ssprim \models A$, if $\exists h^* \in H$ such that
    $\forall h \neq h^*$, $s(h) = s'(h)$. That is, the change affects
    at most one node of the tree.
  \end{enumerate}
  
\end{mydef}
Note that, except for the first two properties, those requirements do
not depend on the outcome. We argue that the first three properties
(Improvement, Subgame Improvement and Laziness) correspond to some
kind of rational behaviours of the players, who seek to improve the
outcome of the game, while performing a minimal amount of changes. On
the other hand, the One Player Property is relevant because such
updates do not allow players to form coalitions to improve their
outcomes. Finally, the Atomicity Property is interesting \textit{per
  se} because it corresponds to some kind of \emph{minimal update},
where a single choice of a single player can be changed at a
time. Because of that, dynamics respecting this Atomicity will be
useful in the rest of the papers to establish general results on
Dynamics.

Based on these properties, we can now define the dynamics that we will
consider in this paper. For all $x\in \{I, SI, L, 1P, A\}$, we define
the $x$-Dynamics as the set of all pairs $(s,s')$ s.t.
$(s,s')\models x$. Intuitively, the $x$-Dynamics is the most
permissive dynamics where the players update their strategies
respecting $x$. For a set $X\subseteq \{I, SI, L, 1P, A\}$, we define
the $X$-dynamics as the intersection of all the $x$-dynamics for
$x\in X$. Throughout the paper, we denote by $\dyn{X}$ the
$X$-Dynamics.

Observe that, following Definition~\ref{def:properties}, any update
satisfying the Atomicity Property also satisfy the One Player
Property. However, no such implication exist in general between the
other properties:
\begin{lemme}\label{lemma : relation between properties}
  Let $s$ and $s'$ be two strategies of a game $G$. Then,
  $(s,s')\models A$ implies that $(s,s')\models 1P$.

  On the other hand, for all $x,y\in \{I, SI, L, 1P, A\}$ s.t.
  $(x,y)\neq (A, 1P)$ and $x\neq y$, there exists a pair of strategies
  $s$ and $s'$ and a game $G$ s.t.: $(s,s')\models x$ and
  $(s,s')\not\models y$.
\end{lemme}
\begin{proof}
  The first point follows from
  Definition~\ref{def:properties}. Indeed, since only one move can be
  made between $s$ and $s'$ (if $(s,s')\models A$), then only one player
  can have updated his strategy.
  
  Let us give a counter example for the lack of implication between
  $I$ and $SI$ (the other cases follow immediately from
  Definition~\ref{def:properties}). Let us consider the game in
  \figurename~\ref{fig:twoplayer}. Then, we claim that
  $(ll,rr)\models I$ but $(ll,rr) \not\models SI$. Indeed, for both
  players, $\outcome{ll} = x \prec \outcome{rr} = z$, but
  $\outcome{ll|_{r}} = y \not\prec_2 \outcome{rr|_r} = z$. On the
  other hand, $(lr,ll) \models SI$ but $(lr,ll)\not\models I$, because
  $\outcome{lr}= x = \outcome{ll}$.
\end{proof}

As a consequence, we obtain direct inclusions between some
dynamics. For instance $\dyn{A}\subseteq\dyn{1P}$ in all games.  In
\cite{LeRouxDynamics}, Le Roux and Pauly consider the so-called
\LazyImp{} Dynamics which corresponds to our $\{I, L,
1P\}$-Dynamics. The underlying idea is to disallow players from making
changes in nodes that are irrelevant (because they will not appear
along the play generated by the profile), while ensuring that the
payoff improves. In \cite{LeRouxDynamics}, Le Roux and Pauly prove
that this dynamics terminates for all games that do not have cyclic
preferences and that the terminal profiles are exactly the Nash
Equilibria.

Examples of graphs associated with dynamics of particular interest (for
the game in \figurename~\ref{fig:twoplayer}) are displayed in
\figurename~\ref{fig:graphs}: the $\{I, L, 1P\}$-Dynamics (or
\LazyImp{} Dynamics) on the left; the $\{SI,A\}$-Dynamics in the
middle (which will be particularly relevant to the discussion in
Section~\ref{sec:subg-impr-dynam}); and the $\{I, L\}$-Dynamics on the
right (which will be particularly relevant in
Section~\ref{sec:cimp-dynamics}).

\begin{figure}
  \begin{center}
	\begin{tikzpicture}[scale=.6]
	 \draw (2,2) node [circle, draw, inner sep=3pt] (s1) {$rr$};
	\draw (4,2) node [circle, draw, inner sep=3pt] (s2) {$rl$};
	\draw (2,0) node [circle, draw, inner sep=3pt] (s3) {$lr$};
	\draw (4,0) node [circle, draw, inner sep=3pt] (s4) {$ll$};

	\draw [->][thick] (s3) to (s1);
	\draw [->][thick] (s1) to (s2);
	\draw [->][thick] (s2) to (s4);
    \end{tikzpicture}      
\hspace*{10mm}
\begin{tikzpicture}[scale=.6]
    
   \draw (2,2) node [circle, draw, inner sep=3pt] (s1) {$rr$};
	\draw (4,2) node [circle, draw, inner sep=3pt] (s2) {$rl$};
	\draw (2,0) node [circle, draw, inner sep=3pt] (s3) {$lr$};
	\draw (4,0) node [circle, draw, inner sep=3pt] (s4) {$ll$};

	\draw [->][thick] (s3) to (s1);
	\draw [->][thick] (s1) to (s2);
	\draw [->][thick] (s2) to (s4);
	\draw [->][thick] (s3) to (s4);
\end{tikzpicture}
\hspace*{10mm}
\begin{tikzpicture}[scale=.6]
    
   \draw (2,2) node [circle, draw, inner sep=3pt] (s1) {$rr$};
	\draw (4,2) node [circle, draw, inner sep=3pt] (s2) {$rl$};
	\draw (2,0) node [circle, draw, inner sep=3pt] (s3) {$lr$};
	\draw (4,0) node [circle, draw, inner sep=3pt] (s4) {$ll$};

	\draw [->][thick] (s3) to (s1);
	\draw [->][thick] (s1) to (s2);
	\draw [->][thick] (s2) to (s4);
	\draw [->][thick] (s4) to (s1);
\end{tikzpicture}

\caption{The graphs corresponding to the \LDyn{} (left),
  \ADyn{} (middle) and \CDyn{} (right),
  respectively, for the game in \figurename~\ref{fig:twoplayer}}
      \label{fig:graphs}
    \end{center}
\end{figure}
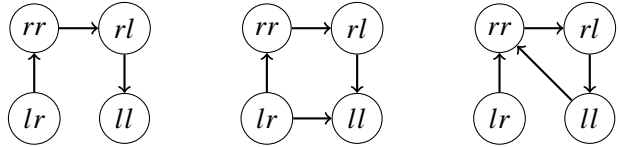

The rest of the paper will be structured as follows: in
Section~\ref{sec:subg-impr-dynam}, we will consider dynamics which are
subsets of the $SI$-Dynamics (like the $\{SI, A\}$-Dynamics). In
Section~\ref{sec:cimp-dynamics}, we will consider dynamics which are
subsets of the $I$-Dynamics, in order to complete the results obtained
by Leroux and Pauly in~\cite{LeRouxDynamics}.


\section{Subgame Improvement Dynamics }\label{sec:subg-impr-dynam}
In this section we will focus on dynamics that respect the
\emph{Subgame Improvement Property} (see
Definition~\ref{def:properties}), i.e., dynamics which are subsets of
the $SI$-Dynamics (note that these dynamics have not been considered
at all in \cite{LeRouxDynamics}). More precisely, we will consider all
the $X$-Dynamics s.t.$\{SI\}\subseteq X\subseteq \{SI, A, 1P\}$. Let
us notice that we do not consider here the $L$ property, because we
argue that there is little interest in the $\{SI,
L\}$-Dynamics. Indeed, let us consider the game in
\figurename~\ref{fig:twoplayer}, with the following preferences instead:
$y\prec_1 x \prec_1 z$ and $y\prec_2z\prec_2 x$. Then, we can
update $ll$ into $rr$ with the $\{SI, L\}$-Dynamics. Observe that in
this update, both players update their strategy and thus agree to form
a coalition to perform it. However, this update is not profitable for
player $2$ who obtains a worse outcome: $z$ instead of $x$. This
example shows that the $\{SI, L\}$-Dynamics yields strategy updates
that are not rational.

The central result of this section is that all those dynamics
terminate when the preferences of the players are
acyclic\footnote{Actually, as we show at the end of the section, in
  the presence of players who have cyclic preferences and play lazily,
  the players who have acyclic preferences are still guaranteed to
  perform a finite number of updates only.} and converge to subgame
perfect equilibria, as stated in the follow
theorem:
\begin{theo}\label{thm : termination and SPE terminal profile}
  Let $N, O, (\prec_i)_{i\in N}$ be respectively a set of players, a
  set outcomes and preferences, and let $X$ be s.t.  $SI \in X
  $. Then, the two following statements are equivalent:
  \begin{enumerate}
  \item In \emph{all} games built over $N, O, (\prec_i)_{i\in N}$ the
    $X$-Dynamics terminates;
  \item The preferences $(\prec_i)_{i\in N}$ are acyclic.
  \end{enumerate}
  Moreover, for all $\{SI\}\subseteq X\subseteq \{SI, 1P, A\}$, the
  set of terminal profiles of the $X$-Dynamics is the set of SPEs of
  the game.
\end{theo}


 

An important consequence of this theorem is that acyclic preferences
form a \emph{sufficient condition} to ensure termination (on SPEs) of
all $X$-dynamics with $SI\in X$. Observe that this condition is very
weak because it constrains only the preferences, and not the structure
(tree) of the game. We argue that this is also very reasonable, as
acyclic preferences allow to capture many if not most rational
preferences\footnote{Although some authors argue that cyclic
  preferences can be realistic, see Larami and Zakir for an example
  \cite[p.~12]{LZ}. } This condition is, however, not necessary:
Theorem~\ref{thm : termination and SPE terminal profile} tells us that
when the preferences are cyclic, the dynamics does not terminate in
\emph{all} games embedding those preferences. Actually, one can find
examples of games with cyclic preferences where the dynamics still
terminate, and examples where they do not. If we consider the game in
\figurename~\ref{fig:twoplayer}, such that preferences of Player~1 are
now $x\prec_1 y \prec_1 z \prec_1 y$ (and Player~2 as before), then
the dynamics will terminate. However, if we let Player~2 have cyclic
preferences too, with: $x\prec_2 y \prec_2 z \prec_2 y$, then he will
infinitely often change his strategy, because he improves the outcome
(according to his preferences) by changing from $y$ to $z$ \emph{and}
from $z$ to $y$.

This section will be mainly devoted to proving Theorem~\ref{thm :
  termination and SPE terminal profile}. Our proof strategy works as
follows. First of all, we establish termination of the
$\{SI\}$-Dynamics when the preferences are acyclic
(Proposition~\ref{prop : SI terminates}). This guarantees that all the
$X$-Dynamics terminate when $SI\in X$, since all these dynamics are
more restrictive. Second, we show that all SPEs appear as terminal
profiles of the $\{SI\}$-Dynamics (Proposition~\ref{prop : SPE eq pour SI}), without guaranteeing, at that point
that all terminal profiles are SPEs. Finally, we show that all the
terminal profiles of the $\{SI, A\}$-Dynamics are SPEs (Proposition~\ref{prop : SPE Adyn}). We then argue,
using specific properties of this dynamics, and relying on
Definition~\ref{def:properties}, that this implies that the set of
terminal profiles of all the dynamics we consider coincide with the
set of SPEs.

\paragraph{The $\{SI\}$-Dynamics} As announced, we start by two
properties of the $\{SI\}$-Dynamics. The first proposition states
that, for a fixed set of preferences, the $\{SI\}$-Dynamics terminates
in all games built on these preferences iff the preferences are
acyclic. The second shows that, in these cases, the SPEs are contained
in the terminal profiles of the $\{SI\}$-Dynamics.

\begin{prop}\label{prop : SI terminates}
  Let $N, O, (\prec_i)_{i\in N}$ be respectively a set of players, a set outcomes
  and preferences. Then, the two following statements are equivalent:
  \begin{inparaenum}[(1)]
  \item In all games $G$ built over $N, O, (\prec_i)_{i\in N}$,
    the $\{SI\}$-Dynamics terminates;
  \item The preferences  $(\prec_i)_{i\in N}$ are acyclic.
  \end{inparaenum}
\end{prop}

\begin{proof}[Sketch of proof]
  Given a cyclic preference, we build a one-player game where all the
  outcomes can be reached from the root (in one step). Since the
  preference is cyclic, the \SIDyn\ does not terminate.
       
  On the other hand, to prove that when preferences are acyclic, the
  dynamics terminates in all games, we make an induction over the
  number of nodes of the game. We consider a game $G$ with $n+1$
  nodes, and in which we have an infinite sequence of strategy
  profiles such that $s^1 \SIdyn s^2 \SIdyn s^3 \SIdyn\ldots$ Then, we
  consider a node $h^* \in H_i$, such that every successor of $h^*$ is
  in $Z$. By the Subgame Improvement Property, if player $i$ changes
  in this place from $x^1$ to $x^2$, he will never come back to $x^1$
  (as preferences are acyclic and he has to improves his payoff at
  this place).  Moreover, as the number of successors of $h^*$ is
  finite, we can consider that, after a finite number of steps in the
  $s^1 \SIdyn s^2 \SIdyn s^3 \SIdyn\ldots$ sequence, Player~$i$ will
  not change anymore in node $h^*$, and definitely choose some
  successor $x^*$ of $h^*$.  Then we can replace $h^*$ by this successor
  $x^*$. We are then in a game with $n$ nodes, in which the dynamics
  terminates by induction hypothesis. The key ingredient is to prove
  that the dynamics coincide in the two games.
\end{proof}

\begin{prop}\label{prop : SPE eq pour SI}
  Let $G$ be a sequential game.  Then, all SPEs of $G$ are terminal
  profiles of $\SIdyn$.
\end{prop}

\begin{proof}[Sketch of proof]
  The property follows directly from the Subgame Improvement Property,
  and the definition of SPEs. Indeed, the SI Property requires to
  improves the outcome in the subgame where players change; and SPEs
  require the players to choose the best response in all
  subgames. Then, when a profile is an SPE, it cannot be updated
  without violating SI.
\end{proof}

\paragraph{The $\{SI, A\}$-Dynamics} We now turn our attention to the
$\{SI,A\}$-Dynamics, and show that all its terminal profiles are
necessarily SPEs. As announced, this result will be sufficient to
establish that the other dynamics we consider here terminate on SPEs
too.

\begin{prop}\label{prop : SPE Adyn}
  Let $G$ be a sequential game. Then, all terminal profiles of $\Adyn$
  are SPEs of $G$.
\end{prop}

\begin{proof}[Sketch of proof]
  The proof is by contradiction: consider a strategy profile $s$ which
  is not an SPE. Then, there is a node where a player does not make
  the best choice for his payoff. He can thus update his choice in
  this node, and this unique update will improve the outcome, it is
  thus allowed in the $\{SI, A\}$-Dynamics. Thus, all profiles which
  are not SPEs are not terminal for this dynamics.
\end{proof}

Although this will not serve to prove Theorem~\ref{thm : termination
  and SPE terminal profile}, we highlight now an interesting property
of the $\{SI,A\}$-Dynamics: roughly speaking, it weakly simulates the
$\{SI\}$-Dynamics, in the sense that, in all games, every update of
the \SIDyn{} can be split into a sequence of updates of \ADyn.

\begin{lemme}\label{lemma : Simulation Adyn} 
  For all sequential games $G$, for all $(s,s')$ s.t.  $s\SIdyn s'$:
  there are $s^1,\ldots, s^k$ s.t.
  $s\Adyn s^1 \Adyn\ldots\Adyn s^k \Adyn s'$.
\end{lemme}
\begin{proof}[Sketch of proof]
  The proof is by induction over $|\Hs{s,s'}|$, the number of changes
  between $s$ and $s'$. First of all, we prove that there is
  $h^*\in \Hs{s,s'}$ such that nothing changes in $H|_{s'(h^*)}$, and
  there is no node higher than $h^*$ where we leave the path towards
  $h^*$ between $s$ and $s'$.  Formally, if $h^* = a^1\cdots a^k$,
  $s|_{s'(h^*)} = s'|_{s'(h^*)}$ and $\nexists h' \in \Hs{s,s'}$,
  $h' = a^1\ldots a^{\ell}$ ($1<\ell<k$), such that $s(h)=a^{\ell+1}$ and
  $\forall \ell<\ell'<k, s(a^1\ldots a^{\ell'})=a^{\ell'+1}$. The point of
  the existence of such a $h^*$ is that, for $s^1$ such that
  $s^1(h^*) = s'(h^*)$ and for $h\neq h^*, s^1(h)=s'(h)$, we have
  $s\Adyn s^1$, and $(s^1, s')\models SI$. Then, as
  $|\Hs{s^1,s'}| = |\Hs{s,s'}|-1$, we conclude by ind. hyp.
\end{proof}

\paragraph{Proof of Theorem~\ref{thm : termination and SPE terminal
    profile}} Equipped with these three propositions, we can now prove
our theorem:

\begin{proof}[Proof of Theorem~\ref{thm : termination and SPE terminal profile}]
  Let us consider a set of players $N$, a set of outcomes $O$,
  preferences $(\prec_i)_{i\in N}$ and $X$ s.t. $SI\in X$. The
  $X$-Dynamics terminates for every game $G$ built over $N, O,
  (\prec_i)_{i\in N}$ if and only if the preferences are acyclic by
  Proposition~\ref{prop : SI terminates}, because ${\dyn{X}} \subseteq
  {\dyn{SI}}$ by definition.

  Next, let us consider $X$ s.t.
  $\{SI\}\subseteq X\subseteq \{SI, A, 1P\}$. By definition, and using
  the fact that Property $A$ implies Property $1P$ (see
  Lemma~\ref{lemma : relation between properties}), we have:
  ${\dyn{SI, A}} \subseteq {\dyn{X}} \subseteq {\dyn{SI}}$. Let $s$ be a
  terminal profile of $\dyn{X}$. Then, it is also terminal in
  $\dyn{SI,A}$ since ${\dyn{SI, A}} \subseteq {\dyn{X}}$. By
  proposition~\ref{prop : SPE Adyn}, $s$ is thus an SPE of $G$. On the
  other hand, let $s'$ be an SPE of $G$. Then by Proposition~\ref{prop
    : SPE eq pour SI}, $s'$ is a terminal node of $\dyn{SI}$. Since
  ${\dyn{X}} \subseteq {\dyn{SI}}$, $s'$ is also a terminal node of
  $\dyn{X}$. Thus, we have shown that all SPEs of $G$ are terminal
  nodes of $G$ and vice-versa.
\end{proof}

\paragraph{Termination in the presence of `cyclic' players} We close
this section by answering the following question: `what happens when
some players have cyclic preferences and some have not?' We call
cyclic the players who have cyclic preferences and show that, although
their presence is sufficient to prevent termination of the whole
dynamics, players with acyclic preferences can still be guaranteed a
bounded number of updates in their choices, provided that the
cyclic players play lazily. Thus, in this case, any infinite
sequence of updates will eventually be made up of updates from the
cyclic players only. This provides some robustness to our
termination result.

Let us consider a set of players $N$ partitioned into the sets $N_c$
and $N_r$ of cyclic and acyclic players respectively; a set $O$ of
outcomes; and preferences $(\prec_i)_{i\in N}$. Let us consider the
dynamics $\Crazydyn$ such that $s\Crazydyn s'$ iff:
\begin{inparaenum}[(i)]
\item either $s\SIdyn s'$ and $d(\Hs{s,s'})\cap N_c = \emptyset$; 
\item or $s\Ldyn s'$ and $d(\Hs{s,s'})\subseteq N_c$\footnote{Observe
    that when $s\Ldyn s'$, the set $d(\Hs{s,s'})$ is necessarily a
    singleton, because of the $1P$-Property.}.
\end{inparaenum}
It means that the acyclic players play according to the \SIDyn, while
the cyclic players have to play according to the \LDyn.  In this case,
we say that the dynamics terminates for acyclic player if there is no
infinite sequence of strategy profiles $(s^k)_{k\in \N}$ such that:
\begin{inparaenum}[(1)]
\item for all $k \in \N$: $s^k \Crazydyn s^{k+1}$
and
\item for all $j\in \N$ there is $k>j$ s.t.
  $d(\Hs{s^k, s^{k+1}}) \cap N_r \neq \emptyset$ (i.e. the acyclic
  players change infinitely often).
\end{inparaenum}
Then:

\begin{prop}\label{prop:cyclic players}
  Let $N=N_r\uplus N_c$, $O$, $(\prec_i)_{i\in N}$ be a set of players
  (with $N_c$ the set of cyclic players), a set of outcomes and a set
  of preferences. Then, the dynamics $\Crazydyn$ terminates for
  acyclic players in all games built over $N, O, (\prec_i)_{i\in N}$.
\end{prop}
 
\begin{proof}[Sketch of proof]
  In \cite[Section~5]{LeRouxDynamics}, Le Roux and Pauly provide an
  alternative proof of the termination for the Lazy Improvement
  Dynamics (here \LDyn). They associate a function per player which
  decrease when the associated player update his strategy, and does
  not otherwise (i.e. when other players update). The termination of
  the \LDyn\ is then a consequence of the decrease of the functions,
  together with the finiteness of strategies (per player).  When
  cyclic players are added, the \LDyn\ still terminates for acyclic
  players as they do not affect the acyclic players.

  In order to obtain the desired result, we need adapt the proof
  of~\cite[Section~5]{LeRouxDynamics}, by introducing a global
  function (i.e. for all acyclic players). This new global function
  has the property to decrease when acyclic players update their
  strategies. Moreover, one also shows that the latter function is not
  affected by the updates of the cyclic players, as their strategies
  follows the \LDyn. This implies that acyclic players do not update
  their strategy an infinite number of time.
\end{proof}

Finally, we note that, if we allow cyclic players to play with the
\SIDyn, the result does not hold.  Consider the game in
\figurename~\ref{fig:twoplayer}, and consider Player~$2$ as the cyclic
one. Then, the graph associated to the \SIDyn\ is represented in
\figurename~\ref{fig : Cyclic and IADyn} (left), where dotted lines
represent updates of the cyclic player (Player~$2$). Clearly, this
graph contains a cycle in which Player~$1$ updates infinitely often.

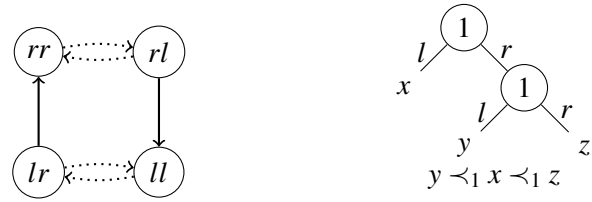
\begin{figure}
  \begin{center}
    \begin{tikzpicture}[scale=.8]
      
      \draw (2,2) node [circle, draw, inner sep=3pt] (s1) {$rr$};
      \draw (4,2) node [circle, draw, inner sep=3pt] (s2) {$rl$};
      \draw (2,0) node [circle, draw, inner sep=3pt] (s3) {$lr$};
      \draw (4,0) node [circle, draw, inner sep=3pt] (s4) {$ll$};
      
      \draw [->][thick] (s3) to (s1);
      \draw [->][thick] (s2) to (s4);
      \draw [->][thick, dotted] (s1) to [bend left=10] (s2);
      \draw [->][thick, dotted] (s2) to [bend left=10] (s1);
      \draw [->][thick, dotted] (s3) to [bend left=10] (s4);
      \draw [->][thick, dotted] (s4) to [bend left=10] (s3);	
    \end{tikzpicture}    
    \hspace*{25mm}
    \begin{tikzpicture}[scale=.8]
      
      \draw (2,3) node [circle, draw, inner sep=3pt] (a) {$1 $};
      \draw (1,2) node (aa) {$x $};
      \draw (3,2.) node  [circle, draw, inner sep=3pt] (ab) {$1 $};
      \draw (2,1) node  (abc) {$y$};
      \draw (4,1) node (aba) {$z $};    
      
      \draw (1.3,2.6)  node {$ l $};
      \draw (2.7,2.6)  node {$ r $};
      \draw (2.3,1.6)  node {$ l $};
      \draw (3.7,1.6)  node {$ r $};
      
      \draw  (a) to (aa);
      \draw  (a) to (ab);
      \draw  (ab) to (abc);
      \draw  (ab) to (aba);
      
      \draw (2.5,0.5) node {$y\prec_1 x\prec_1 z$};
    \end{tikzpicture}
    \caption{Left: graph of the \SIDyn\ for the game of
      \figurename~\ref{fig:twoplayer}, with a cyclic player. Right: an
      example showing that the Terminal profiles of \IADyn{} are not
      always NEs. } 
    \label{fig : Cyclic and IADyn} 
  \end{center}
\end{figure}


\section{Improvement dynamics and coalitions}\label{sec:cimp-dynamics}
While Section~\ref{sec:subg-impr-dynam} was devoted to characterising
the $X$-Dynamics with $SI\in X$, we turn now attention to those where
$I\in X$. Recall that Le Roux and Pauly have studied in
\cite{LeRouxDynamics} the Lazy Improvement Dynamics, which corresponds
to our $\{I, L, 1P\}$-Dynamics and shown that it terminates when the
preferences of the players are acyclic, and terminates to Nash
Equilibria. Their study of the $\{I, L, 1P\}$-Dynamics was motivated
by the fact that less restrictive dynamics (that are still contained
in $\dyn{I}$) do not always terminate, namely the $\{I\}$-Dynamics and
the $\{I, 1P\}$-Dynamics. These results appear as the grey lines in
\tablename~\ref{tab:results-I}.

Our contribution in the present section is to fill in
\tablename~\ref{tab:results-I} by the following results. First, for
all the $X$-Dynamics for $\{I,A\}\subseteq X$, we show that acyclic
preferences guarantee termination, and that the final profiles contain
the Nash Equilibria. 

Second, we consider the $\{I,L\}$-Dynamics which can be regarded as a
\emph{coalition dynamics}, where several players can change their
strategies at the same time to obtain a better outcome for \emph{all
  players} taking part to the coalition.  For example,
in~\figurename~\ref{fig:twoplayer}, the two players can make a
coalition to change from $ll$ to $rr$, as they prefer $z$ to $x$. We
characterise families of games and conditions on the preferences where
termination of the $\{I,L\}$-Dynamics is guaranteed (with the terminal
profiles being exactly the \emph{Strong Nash Equilibria} in the sense
of Aumann \cite{Aumann59}). 

\begin{table}[t]
  \centering
  \caption{Characterisation of the $X$-Dynamics with $I\in X$.
    A \checkmark in the `Games' column means that the hypothesis is
    used to prove the result, and $\times$ that it is not. Hence, the results
    for the classes that are not mentioned
    can be directly deduced, for example, the first line implies
    termination for all $X$-Dynamics with $\{I,A\}\subseteq X$, etc. The
    `Termination' column gives necessary (n.) and/or sufficient (s.)
    conditions to ensure that the dynamics terminates on \emph{all}
    games in the class.} 
  \label{tab:results-I}
  \medskip
  \begin{tabular}{ccccccc}
    \toprule
    $A$   &$L$   &$1P$  &Games &Termination  &Final Profiles  &Reference\\
    \midrule

    \checkmark &$\times$ &$\times$                 & acyclic prefs & \checkmark &$\supseteq$ NEs &
Corollary~\ref{coro : IA terminates}    \\
    \rowcolor[gray]{.9}
    $\times$& $\times$ &\checkmark   &          &$\times$  & not appl.                &\cite{LeRouxDynamics}
    \\
    \multirow{4}{*}{$\times$} & 
    \multirow{4}{*}{\checkmark}&
    \multirow{4}{*}{$\times$}   &swo prefs & prefs can be layered (s.)
                                      &\multirow{4}{*}{$=$ SNEs}
                                                                    &Theorem~\ref{thm : SWO implications}
    \\
    &&&swo prefs &prefs out of pattern (n.)&&Theorem~\ref{thm : SWO implications} \\
    &&&slo prefs &prefs out of pattern (n. \& s.)&&Corollary~\ref{cor:termination-cdyn}\\
    &&&swo prefs, 2 player &prefs out of pattern (n. \& s.)&&Corollary~\ref{cor:termination-cdyn}\\
    \rowcolor[gray]{.9}
    $\times$& \checkmark &\checkmark &acyclic prefs& \checkmark &$=$ NEs       & \cite{LeRouxDynamics}
    \\
    \bottomrule
  \end{tabular}
\end{table}

\subsection{The \IADyn} 
To complete \tablename~\ref{tab:results-I}, we consider now the first
line which represents the $X$-Dynamics with $\{I,A\}\subseteq X$. All
these Dynamics can be considered at once thanks to the next
proposition:
\begin{prop}\label{prop:equiv IADyn and al.}
  All the $X$-Dynamics with
  $\{I, A\}\subseteq X\subseteq \{SI, I, A, L, 1P\}$ are equal.
\end{prop}

\begin{proof}
  To understand these equalities, we must focus on \IADyn. This
  dynamics allows only one update between two profiles, and the
  outcome must be better for the player that has changed his
  strategy. If we want that the outcome of the game changes, the
  atomic move must have occurred along the play induced by the
  strategy. Thus, \IADyn{} verifies the Lazy Property ($L$). Moreover,
  by Lemma~\ref{lemma : relation between properties}, it also verifies
  the One Player Property ($1P$). Finally, as the outcome is improved,
  and only one change has been done, in particular the payoff is
  improved in the subgame rooted at the change. Thus, the \IADyn{} also
  verifies the Subgame Improvement Property ($SI$).
\end{proof}

By Theorem~\ref{thm : termination and SPE terminal profile}, as the
\IADyn{} verifies the Subgame Improvement Property, it terminates for
every game over some $N, O, (\prec_i)_{i\in N}$ if and only if the
preferences are acyclic:
\begin{cor}\label{coro : IA terminates}
  Let $N, O, (\prec_i)_{i\in N}$ be respectively a set of players, a
  set of outcomes and preferences. Then, the two following statements
  are equivalent:
  \begin{inparaenum}[(1)]
  \item in all games $G$ built over $N, O, (\prec_i)_{i\in N}$,
    the $\{I, A\}$-Dynamics terminates;
  \item the preferences  $(\prec_i)_{i\in N}$ are acyclic.
  \end{inparaenum}
\end{cor}

Now that we have established termination, let us turn our attention to
the terminal profiles. It turns out that they contain all Nash
Equilibria of the game:
\begin{prop}\label{prop : IA contains NE}
  Let $G$ be a sequential game. Then all NEs of $G$ are terminal
  profile of $\IAdyn$.
\end{prop}

\begin{proof}
  If $s$ is an NE, no player can improve the outcome alone from $s$,
  hence $s$ is terminal in $\IAdyn$.
\end{proof}
Let us notice that some non-NE profiles can also be terminal profiles
of the~\IADyn. For example, in \figurename~\ref{fig : Cyclic and
  IADyn} (right), $ll$ is a terminal profile, but not an NE.

\subsection{The $\{I,L\}$-Dynamics}
Let us now turn our attention to the $\{I,L\}$-Dynamics. First, we
observe that the terminal profiles of this dynamics are exactly the
Strong Nash Equilibria (see Section~\ref{sec:preliminaries} for the
definition).  We believe the observation is interesting, since, as
already stated, the concept of SNE has sometimes been deemed `too
strong' in the literature. 

As can be seen from \tablename~\ref{tab:results-I}, the
conditions to ensure termination are more involved and require a finer
characterisation of the preferences. We start by discussing these
conditions.

\paragraph{Orders} Let us begin with the definition of strict linear
and weak order. 
A \textbf{strict linear order} over a set $O$ is a total, irreflexive
and transitive binary relation over $O$. This is a natural way to see
orders: for example, the usual orders over $\N$ or $\R$ are strict
linear orders.
A \textbf{strict weak order} over a set $O$ is an irreflexive and
transitive binary relation over $O$ that provides the transitivity of
incomparability. Formally, for $x\neq y\in O$, if $\neg (x<y)$ and
$\neg (y<x)$, we say that $x$ and $y$ are incomparable, and we write
$x\sim y$ (this can happen because a strict weak order is not a total
relation). Then, in a strict weak order, for all $x$, $y$, $z$, we
have that $x\eq y$ and $y\eq z$ implies $x\eq z$.

We write $x\wo y$ if $x<y$ or $x\sim y$, i.e. if $\neg(x\lo
y)$. Sometimes we will denote strict weak orders by $\wo$ to emphasise
the possibility of incomparability, but this does not means that the
relation is reflexive. We argue that strict weak orders are quite
natural to consider in our context. Indeed, the incomparability of two
outcomes for a player reflects the indifference of the player
regarding these outcomes. We can easily imagine two outcomes $x$ and
$y$ such that Player~1 prefers $x$ to $y$ but Player~$2$ has no
preference. This kind of vision justifies the transitivity of
incomparability. Indeed, if a player can neither choose between $x$
and~$y$, nor between $y$ and $z$, it would not seem natural that for
example, he prefers $x$ to $z$. Finally, we say that $\lo'$ is a
\textbf{strict linear extension} of a strict weak order $\lo$ if it is
a strict linear order and $\forall x,y\in O$: $x\lo y$ implies $x<'y$.


\paragraph{Layerability} While the notions of order we have defined
above make sense in our context, they are unfortunately not sufficient
to ensure termination of \CDyn{}. Indeed, considering the game in
\figurename~\ref{fig:twoplayer} and his associated graph with \CDyn{}
in \figurename~\ref{fig:graphs} (right), we can see that, even when
players have strict linear preferences, the dynamics does not
terminate. We thus need to introduce more restrictions on outcomes and
preferences to ensure termination. In \cite{LeRouxPattern}, Le Roux
considers a pattern over outcomes and preferences, and proves that the
absence of this pattern induces some structure on the outcomes that we
call \emph{layerability} (in the case of strict linear
order). Layerability, in turn, can be used to prove termination of
dynamics as we are about to see. Our first task is thus to generalise
the pattern and the definition of layerability of outcomes in this
case in the case of \emph{strict weak orders}.

Let $O$ be a finite set of outcomes, $N$ be a finite set of players
and $(\wo_i)_{i\in N}$ be strict weak orders over~$O$.  For
$(\wo_i)_{i\in N}$ a strict weak order, we say that $(\wo_i)_{i\in N}$
is:
\begin{inparaenum}[(i)]
\item \textbf{out of main pattern} for $O$ if it satisfies:
\begin{align}
   \forall x,y,z\in O,\forall i, j \in N &: \neg (x\lo_i y\lo_i z \text{
  and } y\lo_j z \lo_j x); \label{main pattern}
\end{align}
\item \textbf{out of secondary pattern} for $O$ if it
  satisfies:
\begin{align}
  \forall w,x,y,z\in O,\forall i, j \in N &:
  \neg (w\lo_i x\lo_i y\lo_i z \text{ and } x\eq_j z \lo_j w\eq_j y); \label{sec pattern}
\end{align}
and
\item \textbf{out of pattern} for $O$ if it is out of main and
  secondary pattern.
\end{inparaenum}
Notice that, when $(\lo_i)_{i\in N}$ is a strict linear order, then
$(\lo_i)_{i\in N}$ is out of pattern for $O$ if and only if
$(\lo_i)_{i\in N}$ is out of main pattern for~$O$.

Then, we say that $(\wo_i)_{i\in N}$ \textbf{can be layered} for $O$
if there is a partition $\{O_{\lambda}\}_{\lambda\in I}$ of $O$ (whose
elements are called \emph{layers}) and a strict total order $<$ on $I$
(i.e., the layers are totally ordered) s.t.:
\begin{enumerate}
\item\label{item:8} $\lambda < \mu$ implies that
  $\forall i \in N, \forall x \in O_{\lambda}, y\in O_{\mu}, x\wo_i
  y$; and
\item\label{item:9} $\forall \lambda \in I$, $\forall i,j\in N$,
  $\forall w,x,y,z\in O_{\lambda}$:
  $\neg (x\lo_i y\wedge x\lo_j y\wedge w\lo_i z \wedge z\lo_j w)$.
\end{enumerate}
The intuition between the notion of layer is as follows:
Point~\ref{item:8} tells us that the ordering of the layers is
compatible with the preference relation of \emph{all} players. That
is, if we pick $x$ in some layer $O_\lambda$ and $y$ in some layer
$O_\mu$, with $\lambda<\mu$ (i.e. $O_\mu$ is `better' than
$O_\lambda$), then \emph{all} players will prefer $y$ to $x$. However,
with this point alone in the definition, one could put all the
outcomes in the same layer (i.e., the partition would be trivially
$\{O\}$). Point~\ref{item:9} ensures that the disagreement of the
players on some outcomes is also reflected in the layers. That is, in
all layers $O_\lambda$, we cannot find two players $i$ and $j$ that
agree on a pair of elements $x$ and $y$ from $O_\lambda$ (because they
both prefer $y$ to $x$) but disagree on pair of elements $w$ and $z$
from $O_\lambda$ (because player $i$ prefers $z$ to $w$ but player $j$
prefers $w$ to $z$).  

\begin{table}
  \caption{Four preference orders $\lo_1$, $\lo_2$, $\lo_3$ and  $\lo_4$
    on the set of outcomes $\{u,v,w,x,y,z\}$ that can be layered. }
  \label{table : can be layered}
  
  \centering
  \begin{tabular}{lcccc}
    \toprule
    & $\lo_1$ & $\lo_2$& $\lo_3$&$\lo_4$\\
    \midrule
    \rowcolor[gray]{.9}
    & $u$ & $u$ & $w$ & $u$\\
    \rowcolor[gray]{.9}
    Layer $\nu$ & $v$ & $v$ & $v$ & $v$\\
    \rowcolor[gray]{.9}
    & $w$ & $w$ & $u$ & $w$\\
    \midrule
    Layer $\mu $ & $x$ & $x$ & $x$ & $x$\\
    \midrule
    \rowcolor[gray]{.9}
    & $y$ & $z$ & $z$ & $y$\\
    \rowcolor[gray]{.9}
    \multirow{-2}{*}{Layer $\lambda$} & $z$ & $y$ & $y$ & $z$\\
    \bottomrule
  \end{tabular}
\end{table}

For example, consider the set of outcomes $\{u,v,w,x,y,z\}$ and four
preference orders $\lo_1$, $\lo_2$, $\lo_3$ and $\lo_4$ depicted in
\tablename~\ref{table : can be layered}, where each order corresponds
to a column, and the order of the rows indicate decreasing preference
order (for example, $v\lo_3 w$ since, in column $\lo_3$, $v$ occurs in
the second row, while $w$ occurs in the first row).  With these
orders, the outcomes can be distributed in three layers. In other
words the partition is $\{\{y,z\}, \{x\}, \{u,v,w\}\}$, and the order
on the layers is $\{y,z\} < \{x\} < \{u,v,w\}$. On this example, the
intuition given above can be verified: while the four players do not
agree on the $u$, $v$ and $w$ outcomes, they agree that these outcomes
are all better than $x$, which is always better than $y$ and $z$.

Given this intuition, it is not surprising that the notion of layer
has a strong link with the presence of the pattern that prevents
termination. Indeed, the following result, from \cite{LeRouxPattern}
shows that these notions are \emph{equivalent} in the case of
\emph{strict linear orders}:
\begin{prop}[\cite{LeRouxPattern}]\label{prop : result from Le Roux
    paper}
  Let $O$ be a finite set of outcomes, $N$ a finite set of players and
  $(\lo_i)_{i\in N}$ strict linear orders.  Then $(\lo_i)_{i\in N}$ is
  out of pattern for $O$ if and only if $(\lo_i)_{i\in N}$ can be
  layered for $O$.
\end{prop}

As explained above, we seek to extend this result to strict weak
orders, as this will be ancillary to establishing
termination. Unfortunately, the result does not hold immediately in
this case, as shown by the following example. Consider the followings
preferences : $N = \{1,2,3\}$, $O=\{x,y,z\}$ and
$x\lo_1 y\lo_1 z, y\lo_2z\eq_2 x$ and $y\eq_3 z\lo_3 x$. These
preferences are out of pattern (they satisfy both~\eqref{main pattern}
and \eqref{sec pattern}), but they cannot be layered. However, if we
restrict ourselves to two players, the absence of our main and
secondary pattern is sufficient to ensure that the preferences can be
layered, in the case of strict weak orders, as shown in the next
Proposition. For the sake of clarity, let us denote
$(\wo_i)_{i\in \{1, 2\}}$ (resp. $(\lo_i)_{i\in \{1, 2\}}$) by
$\wo_{1,2}$ (resp. $\lo_{1,2}$). Then:

\begin{prop}\label{prop : 2players Layer extension}
  Let $O$ be a finite set of outcomes, $N=\{1, 2\}$, and $\wo_1, \wo_2$
  strict weak orders. The following statements are equivalent:
\begin{inparaenum}[(1)]
\item $\wo_{1,2}$ is out of pattern for $O$;
\item there exists $\lo'_{1,2}$, a strict linear extension of
  $\wo_{1,2}$ that can be layered for $O$;
\item $\wo_{1,2}$ can be layered for $O$.
\end{inparaenum}
\end{prop}

\paragraph{Termination of the \CDyn} Equipped with these preliminary
results, we can now characterise termination of the \CDyn, as shown in
\tablename~\ref{tab:results-I}.  As we have seen with the example in
\figurename~\ref{fig:twoplayer}, the termination of the \CDyn{} is not
as simple as the termination of \LDyn{} or \ADyn. When we consider
that the preferences of the game are strict weak orders, we will see
that `can be layered' is a necessary condition, and `out of pattern'
is a sufficient condition. However, none of these characterisation are
sufficient and necessary condition.

\begin{theo}\label{thm : SWO implications}
  Let $O$ be a finite set of outcomes, $N$ be a finite set of players
  and $(\wo_i)_{i\in N}$ be strict weak orders.  We have the following
  implications:
  $(\ref{item:1})\Rightarrow (\ref{item:2}) \Rightarrow (\ref{item:3})
  \Rightarrow (\ref{item:4})$, $(\ref{item:2})\nRightarrow (\ref{item:1})$ and
  $(\ref{item:4}) \nRightarrow (\ref{item:3})$ where:
  \begin{enumerate}[(1)]
  \item\label{item:1} $(\wo_i)_{i\in N}$ can be layered for $O$.
  \item\label{item:2} The \CDyn{} terminates in all games built over
    $O$, $N$ and $(\wo_i)_{i\in N}$.
  \item\label{item:3} All games built over $O$, $N$ and
    $(\wo_i)_{i\in N}$ admit an \CEq.
  \item\label{item:4} $(\wo_i)_{i\in N}$ is out of pattern for $O$.
\end{enumerate}
\end{theo}
\begin{proof}[Sketch of proof]
  \fbox{$(\ref{item:1})\Rightarrow (\ref{item:2})$} For this point, the
  idea consists in reducing to a two player game where both players
  play lazily and do not form coalitions. This is the point of the
  proof where we exploit heavily the properties of layers: as
  $(\wo_i)_{i\in N}$ can be layered for $O$, we know that no player or
  coalition of players will make a change in order to reach an outcome
  which is in a lower layer than the current outcome. Indeed, by
  definition of layers, every player prefers every outcome of an upper
  layer to any outcome of a lower layer. Then, we can consider that,
  at some point along the dynamics, we will reach a layer and never
  leave it, because the game is finite. Let us denote it by
  $O_{\lambda}$.

  Moreover, in that layer, we can make two teams of players. Those who
  agree with player $1$, and the others. Indeed, inside a layer, we
  have $\neg (x\lo_i y\wedge x\lo_j y\wedge w\lo_i z \wedge z\lo_j w)$
  for all $i,j\in N$ and for all $w,x,y,z\in O_{\lambda}$.  We can
  then regard the first team as a single player and build a game $G'$
  with two players that never make coalitions, and play lazily. By
  \cite{LeRouxDynamics}, we know that this dynamics terminates for the
  game $G'$ and conclude that the \CDyn{} terminates for the game $G$.
 
  \fbox{$(\ref{item:2})\Rightarrow (\ref{item:3})$} We know that the
  terminal profiles are the SNEs, by definition. Thus, if the dynamics
  terminates, the exists an SNE in the game.

  \fbox{$(\ref{item:3})\Rightarrow (\ref{item:4})$} If $O$ is not out of
  the main pattern, let us consider the game in
  \figurename~\ref{fig:twoplayer}. His associated graph, given by
  \figurename~\ref{fig:graphs} (right), has no terminal node. It means
  that there is no \CEq{}. In the case where $O$ is not out the
  secondary pattern, we consider the game and his associated graph in
  \figurename~\ref{fig : counter examples} (left). The graph has no
  terminal node, so the game has no \CEq{}.

  \fbox{$(\ref{item:2})\nRightarrow (\ref{item:1})$} A counter-example to
  this implication goes as follows. Consider $O=\{x,y,z\}$,
  $N =\{1, 2, 3\}$ and $\wo$ s.t. $y\lo_1 x\eq_1 z$, $z\lo_2 x\lo_2 y$
  and $x\lo_3z\lo_3y$. Clearly, these preferences cannot be layered,
  but we claim that, in all games built over $O$, $N$ and
  $\wo_{1,2,3}$, the \CDyn{} terminates. Indeed, observe that, as
  player $1$ can not make any coalition, he will update his strategy
  only a finite number of times. After that, either player $2$ and
  player $3$ will form a coalition in order to reach $y$, then will
  stop updating their strategy changing; or they will not. In this 
  case, we can consider that they update according to \LDyn{}. Then, 
  by \cite[Theorem 10]{LeRouxDynamics}, this dynamics terminate.

  \fbox{$(\ref{item:4})\nRightarrow (\ref{item:3})$} The game in
  \figurename~\ref{fig : counter examples} (right) is an example where
  the preferences are out of pattern and the graph associated to the
  \CDyn{} has no terminal node, hence the game admits no \CEq{}.
\end{proof}

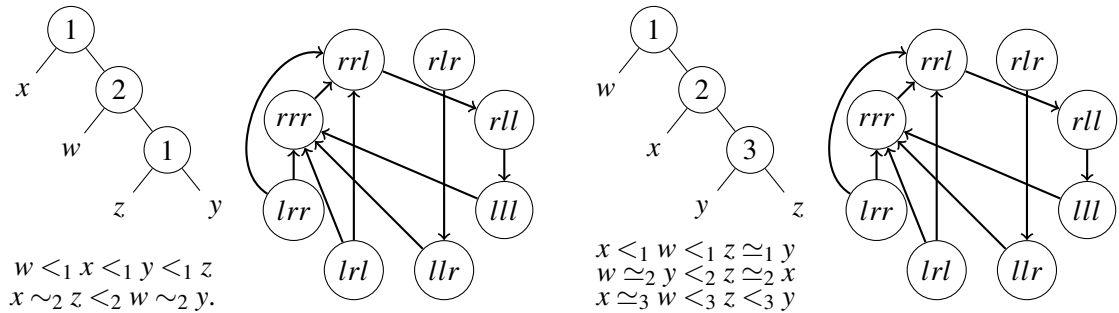
\begin{figure}
  \centering
    \begin{tikzpicture}[scale=.8]
      \begin{scope}
        \draw (0.8,4) node [circle, draw, inner sep=3pt] (e) {$ 1$};
        \draw (0,3) node (w) {$x$};
        
        \draw (1.6,3) node [circle, draw, inner sep=3pt] (a) {$2 $};
        \draw (0.8,2) node (aa) {$w $};
        \draw (2.4,2.) node  [circle, draw, inner sep=3pt] (ab) {$1 $};
        \draw (1.6,1) node  (abc) {$z $};
        \draw (3.2,1) node (aba) {$y $};
        
        \draw  (e) to (a);
        \draw (e) to (w);
        \draw  (a) to (aa);
        \draw  (a) to (ab);
        \draw  (ab) to (abc);
        \draw  (ab) to (aba);
        \draw (1.5,0) node { $w\lo_1 x\lo_1 y\lo_1 z $};
        \draw (1.5,-.5) node {$x\eq_2 z \lo_2 w\eq_2 y$.};
      \end{scope}
      \begin{scope}[xshift=4.5cm]
        \draw (0,2.5) node [circle, draw, inner sep=3pt] (s1) {$rrr$};
        \draw (1,3.5) node [circle, draw, inner sep=3pt](s2) {$rrl$};
        \draw (2.5,3.5) node [circle, draw, inner sep=3pt] (s3) {$rlr$};
        \draw (3.5,2.5) node [circle, draw, inner sep=3pt](s4) {$rll$};
        \draw (0,1) node  [circle, draw, inner sep=3pt] (s5) {$lrr$};
        \draw (1,0) node  [circle, draw, inner sep=3pt](s6) {$lrl$};
        \draw (2.5,0) node [circle, draw, inner sep=3pt](s7) {$llr$};
        \draw (3.5,1) node [circle, draw, inner sep=3pt](s8) {$lll$};
        
        
        \draw [->][thick] (s8) to (s1);
        \draw [->][thick] (s1) to (s2);
        \draw [->][thick] (s2) to (s4);
        \draw [->][thick] (s4) to (s8);
        
        \draw [->][thick] (s3) to (s7);
        \draw [->][thick] (s5) to (s1);
        
        \draw [->][thick,bend left=80] (s5) to (s2);
        \draw [->][thick] (s7) to (s1);
        
        \draw [->][thick] (s6) to (s2);
        \draw [->][thick] (s6) to (s1);
      \end{scope}
    \end{tikzpicture}
    \hspace*{5mm}
   \begin{tikzpicture}[scale=.8]
     \begin{scope}
       \draw (0.8,4) node [circle, draw, inner sep=3pt] (e) {$ 1$};
       \draw (0,3) node (w) {$w$};
       
       \draw (1.6,3) node [circle, draw, inner sep=3pt] (a) {$2 $};
       \draw (0.8,2) node (aa) {$x $};
       \draw (2.4,2.) node  [circle, draw, inner sep=3pt] (ab) {$3 $};
       \draw (1.6,1) node  (abc) {$y $};
       \draw (3.2,1) node (aba) {$z $};
       
       \draw  (e) to (a);
       \draw (e) to (w);
       \draw  (a) to (aa);
       \draw  (a) to (ab);
       \draw  (ab) to (abc);
       \draw  (ab) to (aba);
       
       \draw (1.5,0.3) node { $x\lo_1 w \lo_1 z \simeq_1 y$};
       \draw (1.5,-0.1) node {$w \simeq_2y\lo_2z\simeq_2 x$};
       \draw (1.5,-.5) node {$x\simeq_3w\lo_3z \lo_3 y$ };
     \end{scope}
     \begin{scope}[xshift=4.5cm]
       
       \draw (0,2.5) node [circle, draw, inner sep=3pt] (s1) {$rrr$};
       \draw (1,3.5) node [circle, draw, inner sep=3pt](s2) {$rrl$};
       \draw (2.5,3.5) node [circle, draw, inner sep=3pt] (s3) {$rlr$};
       \draw (3.5,2.5) node [circle, draw, inner sep=3pt](s4) {$rll$};
       \draw (0,1) node  [circle, draw, inner sep=3pt] (s5) {$lrr$};
       \draw (1,0) node  [circle, draw, inner sep=3pt](s6) {$lrl$};
       \draw (2.5,0) node [circle, draw, inner sep=3pt](s7) {$llr$};
       \draw (3.5,1) node [circle, draw, inner sep=3pt](s8) {$lll$};
       
       
       \draw [->][thick] (s8) to (s1);
       \draw [->][thick] (s1) to (s2);
       \draw [->][thick] (s2) to (s4);
       \draw [->][thick] (s4) to (s8);
       
       \draw [->][thick] (s5) to (s1);
       \draw [->][thick,bend left=80] (s5) to (s2);
       \draw [->][thick] (s6) to (s1);
       \draw [->][thick] (s6) to (s2);
       \draw [->][thick] (s7) to (s1);
       \draw [->][thick] (s3) to (s7);
     \end{scope}
   \end{tikzpicture}
   \caption{Two counter-examples with their associated graphs.}
   \label{fig : counter examples}
\end{figure}

Let us now finish by considering our two particular cases, in which we
know that the `layerability' of the order is \emph{equivalent} to the absence
of pattern, namely when preferences are strict linear order
(Proposition~\ref{prop : result from Le Roux paper}), in two player
games (Proposition~\ref{prop : result from Le Roux paper}). Then:

\begin{cor}\label{cor:termination-cdyn}
  Let $O$ be a finite set of outcomes, $N$ be a finite set of players
  and $(\lo_i)_{i\in N}$ be preferences. When \emph{either}
  $N=\{1,2\}$ and $\wo_{1}, \wo_{2}$ are strict weak orders; \emph{or}
  $(\lo_i)_{i\in N}$ are strict linear orders, the following 
  are equivalent:
  \begin{inparaenum}[(1)]
  \item The \CDyn{} terminates in all games built over $O$, $N$ and
    $(\lo_i)_{i\in N}$;
  \item All game built over $O$, $N$ and $(\lo_i)_{i\in N}$ admit an \CEq;
  \item $(\lo_i)_{i\in N}$ is out of pattern for $O$.
  \end{inparaenum}
\end{cor}


\bibliographystyle{eptcs}

\bibliography{biblio}

\end{document}